\documentclass[runningheads]{llncs}

\usepackage[T1]{fontenc}
\usepackage[utf8]{inputenc}
\usepackage{amsmath,amssymb}
\usepackage{graphicx}
\usepackage{tikz}
\usetikzlibrary{calc}
\usepackage{booktabs}
\usepackage[hidelinks]{hyperref}

\newcommand{\dist}{\operatorname{dist}}
\newcommand{\spn}{\operatorname{span}}
\newcommand{\Rone}{\mathbb{R}}
\newcommand{\wt}{w}

\begin{document}

\title{The Facility Advantage in the One-Round Discrete Voronoi Game on a Line}

\titlerunning{The Facility Advantage in the Discrete Voronoi Game on a Line}

\author{Tamal Maharaj}
\authorrunning{T. Maharaj}
\institute{Department of Computer Science,\\
Ramakrishna Mission Vivekananda Educational and Research Institute,\\
Belur Math, Howrah 711202, India\\
\email{tamal@gm.rkmvu.ac.in}}

\maketitle

\begin{abstract}
In the one-round discrete Voronoi game a multiset $V$ of $n$ \emph{voters} on a
line is given; player $\mathcal{P}$ places $k$ facilities, player $\mathcal{Q}$
then places $\ell$, and each voter is won by the nearer facility, ties going to
$\mathcal{P}$. Player $\mathcal{P}$ wins if it keeps at least $n/2$ voters. In the
vocabulary of competitive location this is the absolute $(\ell|k)$-centroid
problem on a path with unit demands, and the responder's problem is the
$(\ell|X_k)$-medianoid, whose closed form on a path --- the sum of the $\ell$
largest of at most $2k$ explicit \emph{marginals} --- is due to Spoerhase and
Wirth. We record this structure, with complete proofs in an appendix, and draw
two consequences that we believe are new.

First, we compute the value of the game against a single responding facility,
$\Gamma_{k,1}(V)$, together with an optimal strategy for $\mathcal{P}$, in
$O(n\log n)$ time for arbitrary positive real demands and every $k$. This
improves the $O(kn\log^2 n)$ bound of Lazar and Tamir for the absolute
$(1|k)$-centroid on a path, and in the Voronoi-game setting compares with the
$O(kn^4)$ general algorithm of de Berg, Kisfaludi-Bak and Mehr.

Second, we study the \emph{facility advantage} $k^\ast(\ell)$, the least $k$ for
which $\mathcal{P}$ wins \emph{every} instance against $\ell$ facilities. We
prove $k^\ast(\ell)\le 2\ell-1$, exhibit instances proving
$k^\ast(\ell)\ge\ell+1$ for $2\le\ell\le6$ (an exact, computer-assisted proof
resting on Spoerhase and Wirth's half-integer discretisation), determine
$k^\ast(1)=1$ and $k^\ast(2)=3$, and show that on uniform instances $k=\ell$
already suffices, so the extremal instances are weighted and $\mathcal{Q}$ wins
them by a single voter. We conjecture $k^\ast(\ell)=\ell+1$ for all $\ell\ge2$.

\end{abstract}

\section{Introduction}
\label{sec:intro}

Voronoi games model competitive facility location: two players place facilities
in a common arena and each is paid by the region, or the set of clients, that is
closer to its own facilities than to the opponent's. The continuous version was
introduced by Ahn et al.~\cite{ahn2004competitive} for a segment and a circle,
and the \emph{one-round} version --- the first player commits to all its
facilities before the second player moves --- by Cheong et
al.~\cite{cheong2004one} for a square; see also Fekete and
Meijer~\cite{fekete2005one}. The economic question goes back to
Hotelling~\cite{hotelling1929stability}.

In the \emph{discrete} variant the players compete for a finite multiset $V$ of
$n$ points, called \emph{voters}. Let $k,\ell\ge 1$. Player $\mathcal{P}$ first
selects a set $P$ of $k$ points of $\Rone^d$, then player $\mathcal{Q}$ selects
a set $Q$ of $\ell$ points. Player $\mathcal{P}$ wins a voter $v\in V$ if
\[
  \dist(v,P)\;\le\;\dist(v,Q),
  \qquad \dist(v,X):=\min_{x\in X}\dist(v,x),
\]
so that ties favour $\mathcal{P}$; write $V[P\succ Q]$ for the multiset of
voters won by $\mathcal{P}$. Player $\mathcal{P}$ \emph{wins the game} if
$|V[P\succ Q]|\ge n/2$. The value of the game is
\[
  \Gamma_{k,\ell}(V)\;:=\;\max_{P\subset\Rone^d,\,|P|=k}\;
        \min_{Q\subset\Rone^d,\,|Q|=\ell}\;\bigl|V[P\succ Q]\bigr| ,
\]
and $\mathcal{P}$ wins iff $\Gamma_{k,\ell}(V)\ge n/2$. For $k=\ell=1$ this is
the classical \emph{majority point} of spatial voting
theory~\cite{mckelvey1976voting,deberg2018plurality}. Banik, Bhattacharya and
Das~\cite{banik2013optimal} studied the line, $d=1$, for $k=\ell$ and gave an
$O(n^{k-\lambda_k})$-time algorithm; de Berg, Kisfaludi-Bak and
Mehr~\cite{deberg2019oneround} gave the first polynomial algorithm for $d=1$,
running in $O(kn^4)$ time for arbitrary $k$ and $\ell$, proved
$\Sigma^P_2$-hardness for $d\ge2$, and asked for ``a simpler (and perhaps
faster) algorithm'' on the line.

\paragraph{The same game in competitive location.} In the operations-research
literature this is the $(r|p)$-centroid problem of Hakimi~\cite{hakimi1983}: a
\emph{leader} places $p$ facilities on a network with weighted demand nodes, a
\emph{follower} answers with $r$, a node patronises the follower only if it is
\emph{strictly} closer, and the leader minimises the follower's market share.
The follower's sub-problem against a fixed leader $X_p$ is the
$(r|X_p)$-medianoid. The one-round discrete Voronoi game on a line is exactly
the \emph{absolute} (facilities anywhere) $(\ell|k)$-centroid problem on a
path whose node weights are the multiplicities of $V$. On a path the follower's
problem is well understood: Megiddo, Zemel and Hakimi~\cite{megiddo1983maxcov}
solve the $(r|X_p)$-medianoid in $O(n)$ time, and Spoerhase and
Wirth~\cite{spoerhase2009centroid} observe that one follower facility inside a
bounded gap of the leader wins exactly an open sub-interval of half the gap's
length, that two facilities clear a gap, that the second facility's increment
never exceeds the first's, and hence that the follower's optimum is the sum of
the $r$ largest of at most $2p$ \emph{marginals}~\cite[Lemma~2.3]{spoerhase2009centroid}.
For the leader's problem they show that the absolute $(r,p)$-centroid on a path
is NP-hard when $r$, $p$ and arbitrary demand weights are part of the input
(the Voronoi-game algorithms are polynomial in the \emph{unary} total demand
$n$, which is consistent), that the discrete version is solvable in $O(pn^4)$
time by a $k$-sum shortest path reduction, and that the $(1,p)$-centroid on a
tree is polynomial. Lazar and Tamir~\cite{lazar2013improved} made the latter
strongly polynomial and, on a path, obtained $O(pn\log^2 n)$ for the absolute
$(1,p)$-centroid with positive real weights~\cite[Theorem~6.2]{lazar2013improved}.

Although de Berg et al.\ cite~\cite{spoerhase2009centroid}, the two literatures
have developed largely separately, and the explicit structure of the responder
has not, to our knowledge, been exploited on the Voronoi-game side. This paper
does so.

\subsection*{Our contribution}

\begin{enumerate}
\item \textbf{The structure of the responder (Section~\ref{sec:structure}).}
  We restate the gap decomposition of the responder's optimum in the precise
  form we need --- with the role of ties, of open windows, and of voter blocks
  of span less than half a gap made explicit --- with complete proofs in
  Appendix~\ref{app:proofs}. Nothing in this section is new in substance; it
  is included because every later argument manipulates these marginals exactly.

\item \textbf{$\Gamma_{k,1}$ in $O(n\log n)$ time (Section~\ref{sec:ell1}).}
  Against a single responding facility the leader's problem is a min--max over
  the marginals. Following the \emph{$W$-bounding set} greedy of Spoerhase,
  Wirth, Lazar and Tamir, we place the leader's points left to right, each as
  far right as a target value permits. Our contribution is to evaluate one
  sweep of this greedy in $O(n)$ time for \emph{all} $k$ at once, by
  precomputing the minimal heavy voter blocks; Lazar and Tamir spend
  $O(n\log n)$ per placed point. With the same outer search over the $O(n^2)$
  candidate values this yields $O(n\log n)$ total time for positive real
  demands (Theorem~\ref{thm:ell1}), improving $O(kn\log^2 n)$. In the
  Voronoi-game vocabulary: $\Gamma_{k,1}(V)$ and an optimal $P$ in
  $O(n\log n)$ time against $O(kn^4)$.

\item \textbf{The facility advantage (Section~\ref{sec:advantage}).} Let
  $k^\ast(\ell)$ be the least $k$ such that $\Gamma_{k,\ell}(V)\ge n/2$ for
  \emph{every} finite multiset $V\subset\Rone$. Banik, De Carufel, Maheshwari
  and Smid~\cite{banik2016epsnets} proved via $\epsilon$-nets that in two and
  three dimensions some constant factor $k=c\ell$ suffices, without determining
  $c$. On a line we prove $k^\ast(\ell)\le2\ell-1$ (Theorem~\ref{thm:upper}),
  give instances proving $k^\ast(\ell)\ge\ell+1$ for $2\le\ell\le6$
  (Theorem~\ref{thm:lower}), settle $k^\ast(1)=1$ and $k^\ast(2)=3$, and prove
  that on \emph{uniform} instances $k=\ell$ suffices
  (Proposition~\ref{prop:uniform}) --- so the extremal instances are weighted
  and the lower bound is an indivisibility obstruction, $\mathcal{Q}$ winning
  by exactly one voter. The lower bounds are computer-assisted but exact: the
  half-integer discretisation of Spoerhase and Wirth~\cite[Theorem~3.4]{spoerhase2009centroid}
  makes the continuous maximisation over $P$ a finite enumeration.

\item \textbf{Verification (Section~\ref{sec:exp}).} Every structural and
  algorithmic statement was checked by exhaustive computation against a blind
  brute-force implementation of the game definition.
\end{enumerate}

Throughout, $d=1$ and $V=\{v_1\le\dots\le v_n\}$; we identify a point with its
coordinate. For a sub(multi)set $S$ of voters, $\wt(S)$ is its cardinality
counted with multiplicity (its \emph{weight}), and $\spn(S):=\max S-\min S$.

\section{The Structure of the Responder}
\label{sec:structure}

Fix a strategy $P=\{p_1<\dots<p_k\}$ of $\mathcal{P}$. The points of $P$ cut the
line into $k+1$ \emph{gaps}:
\[
  G_0:=(-\infty,p_1),\qquad
  G_i:=(p_i,p_{i+1})\ (1\le i\le k-1),\qquad
  G_k:=(p_k,+\infty).
\]
We call $G_0,G_k$ the \emph{outer} gaps and $G_1,\dots,G_{k-1}$ the \emph{inner}
gaps, and write $V_G:=V\cap G$. Voters lying exactly on a point of $P$ are won by
$\mathcal{P}$ whatever $\mathcal{Q}$ does; we call them \emph{absorbed} and write
$A(P)$ for their multiset. For a single point $q$ let
$W(q):=\{v\in V:\ |v-q|<\dist(v,P)\}$ be the voters $q$ takes from
$\mathcal{P}$; the multiset won by $\mathcal{Q}$ is $\bigcup_{q\in Q}W(q)$.
The statements of this section are due to Spoerhase and
Wirth~\cite{spoerhase2009centroid}, in the discrete setting; the short proofs
for the absolute setting used here are collected in Appendix~\ref{app:proofs}.

\begin{lemma}[Locality]
\label{lem:locality}
Let $q\in\Rone\setminus P$ and let $G$ be the gap of $P$ containing $q$. Then
$W(q)\subseteq G$, and
\begin{enumerate}
 \item[(a)] if $G=G_0$ then $W(q)=V\cap\bigl(-\infty,\tfrac{q+p_1}{2}\bigr)$;
 \item[(b)] if $G=G_k$ then $W(q)=V\cap\bigl(\tfrac{p_k+q}{2},\infty\bigr)$;
 \item[(c)] if $G=G_i=(p_i,p_{i+1})$ then
       $W(q)=V\cap\bigl(\tfrac{p_i+q}{2},\,\tfrac{q+p_{i+1}}{2}\bigr)$,
       an \emph{open} interval of length $\tfrac{1}{2}|G_i|$.
\end{enumerate}
\end{lemma}

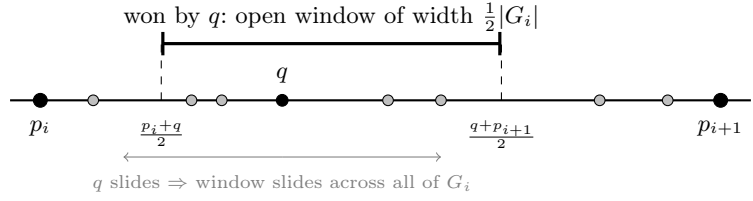
\begin{figure}[t]
\centering
\begin{tikzpicture}[x=1cm,y=1cm]
  \draw[thick] (-0.4,0) -- (9.4,0);
  \foreach \x in {0,9} \filldraw (\x,0) circle (2.6pt);
  \node[below] at (0,-0.15) {$p_i$};
  \node[below] at (9,-0.15) {$p_{i+1}$};
  \filldraw (3.2,0) circle (2.2pt);
  \node[above] at (3.2,0.12) {$q$};
  \draw[very thick,|-|] (1.6,0.75) -- (6.1,0.75);
  \node at (3.85,1.1) {\small won by $q$: open window of width $\tfrac{1}{2}|G_i|$};
  \draw[dashed] (1.6,0) -- (1.6,0.75);
  \draw[dashed] (6.1,0) -- (6.1,0.75);
  \node[below] at (1.6,-0.15) {\scriptsize $\frac{p_i+q}{2}$};
  \node[below] at (6.1,-0.15) {\scriptsize $\frac{q+p_{i+1}}{2}$};
  \foreach \x in {0.7,2.0,2.4,4.6,5.3,7.4,8.3}
     \draw[fill=black!25] (\x,0) circle (2.0pt);
  \draw[->,gray] (3.2,-0.75) -- (1.1,-0.75);
  \draw[->,gray] (3.2,-0.75) -- (5.3,-0.75);
  \node[gray] at (3.2,-1.1) {\scriptsize $q$ slides $\Rightarrow$ window slides across all of $G_i$};
\end{tikzpicture}
\caption{Lemma~\ref{lem:locality}(c): one facility of $\mathcal{Q}$ inside a
bounded gap wins an open window of width exactly half the gap, whose position
follows $q$~\cite[\S2.1]{spoerhase2009centroid}.}
\label{fig:window}
\end{figure}

As $q$ ranges over $(p_i,p_{i+1})$ the left endpoint $\tfrac{p_i+q}{2}$ ranges
over $\bigl(p_i,\tfrac{p_i+p_{i+1}}{2}\bigr)$, so the window sweeps the whole gap
(Fig.~\ref{fig:window}). We make the attainable sets precise; the one subtlety is
a voter exactly at the midpoint of the gap.

\begin{lemma}[Free sliding]
\label{lem:window}
Let $G=(a,b)$ be an inner gap and $\lambda:=\tfrac{1}{2}(b-a)$. Every $W(q)$,
$q\in G$, equals $V\cap I$ for an open interval $I\subseteq(a,b)$ of length
$\lambda$. Conversely, for every open interval $I\subseteq(a,b)$ of length
$\lambda$ there is a $q\in G$ with $W(q)\supseteq V\cap I$, with equality
whenever the left endpoint of $I$ lies strictly between $a$ and $a+\lambda$.
\end{lemma}

\begin{lemma}[Blocks]
\label{lem:block}
Let $G=(a,b)$ be an inner gap, $\lambda=\tfrac{1}{2}(b-a)$, and
\[
   h(G)\;:=\;\max\bigl\{\wt(B)\ :\ B\subseteq V_G \text{ a contiguous block with }
                        \spn(B)<\lambda \bigr\},
\]
where a \emph{contiguous block} is $V\cap J$ for an interval $J$, and $h(G)=0$ if
$V_G=\emptyset$. Then the maximum weight $\mathcal{Q}$ wins in $G$ with a single
facility is exactly $h(G)$.
\end{lemma}

\begin{lemma}[Clearing a gap; concavity]
\label{lem:clear}
One facility placed in an outer gap $G$ wins all of $V_G$; two facilities placed
in an inner gap $G$ win all of $V_G$; further facilities gain nothing. Moreover,
for every inner gap, $h(G)\ge\tfrac12\wt(V_G)$, i.e.\ $h(G)\ge\wt(V_G)-h(G)$.
\end{lemma}

Define the \emph{marginal multiset} of $P$:
\[
  M(P)\ :=\ \bigl\{\wt(V_{G_0}),\ \wt(V_{G_k})\bigr\}\ \cup
  \bigcup_{i=1}^{k-1}\bigl\{\,h(G_i),\ \wt(V_{G_i})-h(G_i)\,\bigr\},
\]
with zero entries discarded; $|M(P)|\le 2k$.

\begin{theorem}[Gap decomposition~{\cite[Lemma~2.3]{spoerhase2009centroid}}]
\label{thm:decomp}
For every $P$ with $|P|=k$ and every $\ell\ge1$,
$\max_{|Q|=\ell}\bigl|V\setminus V[P\succ Q]\bigr|=\sum_{j=1}^{\ell}\mu_j$,
where $\mu_1\ge\mu_2\ge\cdots$ is $M(P)$ in non-increasing order (padded with
zeros). Given $V$ and $P$ sorted, this value and an optimal $Q$ are computed in
$O(n+k)$ time.
\end{theorem}

Two consequences match remarks in~\cite{deberg2019oneround}: if $\ell\ge2k$ then
$\mathcal{Q}$ wins every non-absorbed voter, and $\mathcal{P}$ always keeps the
$k$ heaviest locations by occupying them. A third will be used repeatedly.

\begin{lemma}[More facilities never hurt]
\label{lem:monok}
For every $P$, every $x\notin P$ and every $\ell$,
$\min_{|Q|=\ell}|V[(P\cup\{x\})\succ Q]|\ \ge\ \min_{|Q|=\ell}|V[P\succ Q]|$.
In particular $\Gamma_{k+1,\ell}(V)\ge\Gamma_{k,\ell}(V)$.
\end{lemma}

\begin{proof}
The point $x$ splits one gap $G$ into $G',G''$ and leaves the other value
functions unchanged; voters at $x$ become absorbed, so
$\wt(V_{G'})+\wt(V_{G''})\le\wt(V_G)$. Every block admissible for $G'$ or $G''$
(Lemma~\ref{lem:block}) lies in $G$ with span below the smaller threshold, so
$h(G'),h(G'')\le h(G)$ when $G$ is inner; when $G$ is outer each single-facility
value is at most $\wt(V_G)=\sigma_G(1)$. Hence for all $t',t''\ge0$,
$\sigma_{G'}(t')+\sigma_{G''}(t'')\le\sigma_G(t'+t'')$: for $t'+t''=1$ the single
facility earns at most $\sigma_G(1)$, and for $t'+t''\ge2$ at most
$\wt(V_{G'})+\wt(V_{G''})\le\sigma_G(2)$. Every allocation against $P\cup\{x\}$
is thus dominated by the merged allocation against $P$. \qed
\end{proof}

\section{$\Gamma_{k,1}$ in $O(n\log n)$ Time}
\label{sec:ell1}

Fix $\ell=1$. In this section the voters may carry arbitrary positive real
weights: $V$ consists of $m$ distinct positions $u_1<\dots<u_m$ with weights
$\omega_1,\dots,\omega_m>0$, $\wt(S)$ is the total weight, and $n:=\wt(V)$; the
unit-weight game is the special case $\omega_i\in\mathbb{Z}_{>0}$, and we write
$n$ for $m$ in running times since $m\le n$ there. Nothing in
Section~\ref{sec:structure} used integrality. By Theorem~\ref{thm:decomp} the
responder takes the single largest marginal, so
\[
  \Gamma_{k,1}(V)\;=\;n-\tau^\ast,
  \qquad
  \tau^\ast:=\min_{|P|=k}\ \max M(P),
\]
a min--max problem for $\mathcal{P}$. For a target $t\ge0$ call $P$
\emph{$t$-safe} if
\[
   \wt(V_{G_0})\le t,\qquad \wt(V_{G_k})\le t,\qquad
   h(G_i)\le t \quad (1\le i\le k-1).
\]
Then $\tau^\ast\le t$ iff some $P$ with $|P|\le k$ is $t$-safe (by
Lemma~\ref{lem:monok} unused points may be added). In centroid terms a minimal
$t$-safe set is Spoerhase and Wirth's \emph{$W$-bounding
set}~\cite{spoerhase2009centroid}, and the left-to-right greedy below is theirs
and Lazar and Tamir's~\cite[\S6.4]{lazar2013improved}; what is new is the
$O(n)$ evaluation of a whole sweep in Theorem~\ref{thm:ell1}.

\begin{lemma}[Monotonicity]
\label{lem:mono}
Fix $t$ and $a\in\Rone$ and let $N_t(a):=\sup\{b>a:\ h((a,b))\le t\}$. Then
\begin{enumerate}
\item[(i)] $h((a,b))$ is non-decreasing in $b$ and non-increasing in $a$;
\item[(ii)] the supremum is attained whenever $N_t(a)<\infty$;
\item[(iii)] $N_t$ is non-decreasing.
\end{enumerate}
\end{lemma}

\begin{proof}
(i) Enlarging $(a,b)$ enlarges both the set of voters it contains and the
threshold $\lambda$, so it can only enlarge the family of admissible blocks.

(ii) By Lemma~\ref{lem:block}, $h((a,b))\le t$ iff no contiguous block $B$ with
$\wt(B)>t$ is simultaneously \emph{contained} in $(a,b)$ and \emph{coverable},
i.e.\ iff every such block $B$ with $\min B>a$ satisfies
\[
  \max B\ \ge\ b \qquad\text{or}\qquad \spn(B)\ \ge\ \tfrac12(b-a),
\]
equivalently $b\le\max\bigl\{\max B,\ a+2\spn(B)\bigr\}$. Both escapes matter:
a heavy block is harmless either because the window is too narrow to cover it
\emph{or} because it does not lie inside the gap --- a single heavy voter, of
span $0$, is neutralised only by the second. This is a finite conjunction of
closed conditions on $b$, so the feasible set is closed on the right.

(iii) Let $a\le a'$. If $N_t(a)\le a'$ there is nothing to show. Otherwise
$b:=N_t(a)>a'$ and by (ii) and (i), $h((a',b))\le h((a,b))\le t$, so
$b\le N_t(a')$. \qed
\end{proof}

\begin{lemma}[Greedy]
\label{lem:greedy}
Fix $t$, let $p_1:=\sup\{x:\ \wt(V\cap(-\infty,x))\le t\}$ and
$p_{j+1}:=N_t(p_j)$ (stopping if some $N_t(p_j)=\infty$). A $t$-safe strategy
with at most $k$ points exists iff the process stops within $k$ steps or
$\wt(V\cap(p_k,\infty))\le t$.
\end{lemma}

\begin{proof}
Sufficiency is by construction ($p_1$ is the position of the first voter at
which the prefix weight exceeds $t$, or $+\infty$). For necessity let
$P'=\{p'_1<\dots<p'_{m'}\}$, $m'\le k$, be $t$-safe; we show $p'_j\le p_j$ by
induction. $p'_1\le p_1$ by definition of $p_1$. If $p'_j\le p_j$ then
$h((p'_j,p'_{j+1}))\le t$ gives $p'_{j+1}\le N_t(p'_j)\le N_t(p_j)=p_{j+1}$ by
Lemma~\ref{lem:mono}(iii). Hence $p'_{m'}\le p_k$ and
$\wt(V\cap(p_k,\infty))\le\wt(V\cap(p'_{m'},\infty))\le t$. \qed
\end{proof}

\begin{theorem}
\label{thm:ell1}
For voters with positive real weights and every $k\ge1$, $\Gamma_{k,1}(V)$ and
an optimal strategy for $\mathcal{P}$ can be computed in $O(n\log n)$ time
($O(n)$ per candidate value $t$, for $O(\log n)$ candidates).
\end{theorem}

\begin{proof}
\emph{One safety test in $O(n)$ time.} Fix $t$. Compute, by a two-pointer sweep,
$j(i):=\min\{j\ge i:\ \omega_i+\dots+\omega_j>t\}$ (or $\infty$), the shortest
\emph{heavy} block starting at $u_i$. If a block $B'$ is heavy then the minimal
block $B=[u_i,u_{j(i)}]$ with the same left endpoint satisfies $B\subseteq B'$,
so whenever $B'$ is contained in the gap and coverable, so is $B$; hence it
suffices to impose Lemma~\ref{lem:mono}(ii) for these $m$ blocks:
\[
  N_t(a)\;=\;\min_{\,u_i>a}\ \max\bigl\{\,u_{j(i)},\ \ a+2\,(u_{j(i)}-u_i)\,\bigr\},
  \qquad \min\emptyset=\infty.
\]
Evaluate it by scanning $i$ upwards from the first voter right of $a$, keeping
the running minimum $b$ and stopping as soon as $u_i\ge b$ (a block starting at
or beyond $b$ is not inside the gap). Every candidate is at least $u_{j(i)}\ge
u_i$, so the minimiser $i^\ast$ has $u_{i^\ast}\le N_t(a)$; positions are
scanned in increasing order, so once $i^\ast$ is passed the running minimum is
$N_t(a)$ and the scan stops at the first voter $\ge N_t(a)$. The scan therefore
touches only voters in $(a,N_t(a)]$, plus one. The gaps produced by the greedy
are disjoint, so a whole sweep of Lemma~\ref{lem:greedy} costs $O(m+k)$; it
also reports the number of points used, hence answers the test for all $k$
simultaneously.

\emph{The outer search.} Every marginal is the weight of a contiguous block, so
$\tau^\ast\in\mathcal{W}:=\{\,\Omega_j-\Omega_i:\ 0\le i\le j\le m\,\}$ with
$\Omega_i:=\omega_1+\dots+\omega_i$; the map $(i,j)\mapsto\Omega_j-\Omega_i$ is an
$X+Y$ matrix, in which the element of any prescribed rank is found in $O(m)$
time by Frederickson and Johnson~\cite{frederickson1982selection}. Safety is
monotone in $t$, so binary search over ranks needs $O(\log m^2)=O(\log m)$
selections and tests, i.e.\ $O(m\log m)$ time in all, after sorting. (In the
unit-weight game one may simply binary-search $t\in\{0,\dots,n\}$.) The greedy's
witness at $\tau^\ast$ is an optimal $P$. \qed
\end{proof}

Lazar and Tamir~\cite[Theorem~6.2]{lazar2013improved} obtain $O(kn\log^2 n)$ for
the same problem; their test locates each of the $k$ leader points by a binary
search costing $O(n\log n)$, whereas the sweep above places all of them in one
$O(n)$ pass. On the Voronoi-game side the only previous bound is the $O(kn^4)$
general algorithm of~\cite{deberg2019oneround}.

\section{The Facility Advantage $k^\ast(\ell)$}
\label{sec:advantage}

We return to unit weights (multisets) and ask how large a numerical advantage
the first player needs in order to win unconditionally.

\begin{definition}
$k^\ast(\ell):=\min\{k\ :\ \Gamma_{k,\ell}(V)\ge n/2$ for every finite multiset
$V\subset\Rone\}$.
\end{definition}

By Lemma~\ref{lem:monok} the property defining $k^\ast$ is monotone in $k$, so
an instance on which $\mathcal{P}$ loses with $k$ facilities proves
$k^\ast(\ell)>k$.

\subsection{An upper bound}

\begin{lemma}[Balanced placement]
\label{lem:balanced}
For every multiset $V$ of $n$ voters and every $k\ge1$, $\mathcal{P}$ can place
$k$ points so that every gap contains at most
$B:=\bigl\lceil (n-k)/(k+1)\bigr\rceil$ voters.
\end{lemma}

\begin{proof}
Let $v_1\le\dots\le v_n$ and $r_0:=0$. For $j=1,\dots,k$: if $r_{j-1}+B\ge n$,
place the remaining points to the right of $v_n$ and stop --- the last
non-empty gap then holds $n-r_{j-1}\le B$ voters. Otherwise put
$p_j:=v_{\,r_{j-1}+B+1}$ and $r_j:=|\{i:v_i\le p_j\}|$. The voters strictly
between $p_{j-1}$ and $p_j$ (or left of $p_1$) are among
$v_{r_{j-1}+1},\dots,v_{r_{j-1}+B}$, so at most $B$; and $r_j\ge r_{j-1}+B+1$.
Hence $r_k\ge k(B+1)$ and the right outer gap holds
$n-r_k\le n-k-kB\le B$, as $B\ge(n-k)/(k+1)$. \qed
\end{proof}

\begin{theorem}
\label{thm:upper}
$k^\ast(\ell)\le 2\ell-1$ for every $\ell\ge1$.
\end{theorem}

\begin{proof}
Let $k=2\ell-1$ and take $P$ from Lemma~\ref{lem:balanced}, so every gap holds at
most $B=\lceil (n-2\ell+1)/(2\ell)\rceil$ voters. Every entry of $M(P)$ is at
most the weight of its gap, hence at most $B$, so by Theorem~\ref{thm:decomp}
$\mathcal{Q}$ wins at most $\ell B$. Write $n=2\ell s+r$, $0\le r<2\ell$; then
$(n-2\ell+1)/(2\ell)=s-1+(r+1)/(2\ell)$ with $0<r+1\le2\ell$, so $B=s$ and
$\mathcal{Q}$ wins at most $\ell s\le n/2$. \qed
\end{proof}

For $\ell=1$ this gives $k^\ast(1)=1$: a median of $V$ is a majority point on a
line.

\subsection{Uniform instances are not extremal}

Let $U_m$ be the \emph{uniform} instance: $m$ voters of multiplicity $1$ in
arithmetic progression.

\begin{proposition}
\label{prop:uniform}
Let $k\ge\ell$ and suppose $2k$ divides $m$. Then $\Gamma_{k,\ell}(U_m)\ge m/2$,
attained with all $2k$ marginals equal to $m/(2k)$.
\end{proposition}

\begin{proof}
Normalise the voters to $1,\dots,m$ and put $a:=m/(2k)$. Place the $k$ points of
$P$ at \emph{midpoints} $j+\tfrac12$ so that the outer gaps contain $a$ voters
each and every inner gap contains $2a$; this is possible since
$2a+(k-1)2a=m$, and nothing is absorbed. An inner gap has width $2a$, so
$\lambda=a$; a block of $j$ consecutive voters has span $j-1<a$ iff $j\le a$,
hence $h=a$ and the gap contributes the marginals $a,a$. Each outer gap
contributes $a$. So $M(P)$ consists of $2k$ copies of $a$, and $\mathcal{Q}$
wins $\ell a\le ka=m/2$. \qed
\end{proof}

The divisibility hypothesis is an artefact of the argument: for every
$\ell\le5$, $k\ge\ell$ and $m\le30$ we verified by exhaustive computation that
$\mathcal{P}$ still wins $U_m$, trading midpoint placements against placements
on voters as the residue of $m$ modulo $2k$ demands. Consequently any instance
witnessing $k^\ast(\ell)>\ell$ must use non-trivial multiplicities; and since
the bound $2ka\le m$ is tight, the margin by which $\mathcal{Q}$ can win is
small. The lower bound below is of exactly this type.

\subsection{A lower bound}

For $\ell\ge2$ define the instance $F(\ell)$ on the $2\ell+2$ integer positions
$0,1,\dots,2\ell+1$ with multiplicities
\[
  \bigl(\underbrace{2,1,1,2,1,2}_{\text{positions }0..5},
        \underbrace{2,2,\dots,2}_{2\ell-4\ \text{positions}}\bigr),
  \qquad n=4\ell+1 ;
\]
three voters of multiplicity $1$, $2\ell-1$ of multiplicity $2$. For $\ell=2$
this is the $9$-voter instance $\{0,0,1,2,3,3,4,5,5\}$.

\begin{theorem}
\label{thm:lower}
For $2\le\ell\le6$, $\Gamma_{\ell,\ell}\bigl(F(\ell)\bigr)=2\ell<n/2$, hence
$k^\ast(\ell)\ge\ell+1$. In particular $k^\ast(2)=3$ and $4\le k^\ast(3)\le5$.
\end{theorem}

The proof is computer-assisted but exact, thanks to the following
discretisation, which Spoerhase and Wirth proved for arbitrary graphs.

\begin{lemma}[{Spoerhase and Wirth~\cite[Theorem~3.4]{spoerhase2009centroid}}]
\label{lem:disc}
Let $V\subset\mathbb{Z}$. Then for all $k,\ell$ the value $\Gamma_{k,\ell}(V)$ is
attained by some $P\subset\tfrac12\mathbb{Z}$.
\end{lemma}

On a line this can also be read off Theorem~\ref{thm:decomp}: $h(G)$ depends on
the width $w$ of an inner gap with integer contents only through
$\max\{d\in\mathbb{Z}:2d<w\}$, which never increases when each non-integer
$p_j$ is moved to $\lfloor p_j\rfloor+\tfrac12$, while gap contents and absorbed
voters are unchanged; replacing $(h,\wt-h)$ by $(h',\wt-h')$ with
$\wt/2\le h'\le h$ cannot increase the sum of the $\ell$ largest marginals.

\begin{proof}[of Theorem~\ref{thm:lower}]
The upper bound $\Gamma_{\ell,\ell}(F(\ell))\le2\ell$ was established by
exhaustive enumeration of all $k$-subsets of
$\tfrac12\mathbb{Z}\cap[-1,2\ell+2]$, each evaluated by
Theorem~\ref{thm:decomp}. By Lemma~\ref{lem:disc} half-integers suffice, and a
point more than one unit outside the voter range is dominated: moving it
inwards to the nearest half-integer outside the range keeps every gap's
contents and shrinks its inner gap, which by Lemma~\ref{lem:mono}(i) cannot
increase any marginal (colliding points are re-placed using
Lemma~\ref{lem:monok}). As an independent check the enumeration was repeated
over the finer grids $\tfrac14\mathbb{Z},\tfrac18\mathbb{Z},\tfrac1{12}\mathbb{Z}$
and $\tfrac1{16}\mathbb{Z}$ with identical results. Since $n=4\ell+1$ is odd,
$\Gamma_{\ell,\ell}=2\ell<2\ell+\tfrac12=n/2$: $\mathcal{Q}$ wins against
$k=\ell$, and by Lemma~\ref{lem:monok} against every $k\le\ell$, so
$k^\ast(\ell)\ge\ell+1$. For $2\le\ell\le5$ the same enumeration gives
$\Gamma_{\ell+1,\ell}(F(\ell))=2\ell+2>n/2$, so $F(\ell)$ does not rule out
$k^\ast(\ell)=\ell+1$. With Theorem~\ref{thm:upper}: $k^\ast(2)=3$ and
$4\le k^\ast(3)\le5$. \qed
\end{proof}

\begin{table}[t]
\centering
\caption{The facility advantage on a line. For $\ell=1$ the lower bound is
trivial; for $\ell\ge2$ it is Theorem~\ref{thm:lower}. Exhaustive search over
the instance family of Section~\ref{sec:exp} found no instance defeating
$k=\ell+1$, which is why we expect the middle row to be the truth
(Conjecture~\ref{conj}).}
\label{tab:advantage}
\setlength{\tabcolsep}{9pt}
\begin{tabular}{lcccccc}
\toprule
$\ell$ & $1$ & $2$ & $3$ & $4$ & $5$ & $6$\\
\midrule
lower bound (Thm.~\ref{thm:lower}) & $1$ & $3$ & $4$ & $5$ & $6$ & $7$\\
upper bound $2\ell-1$ (Thm.~\ref{thm:upper}) & $1$ & $3$ & $5$ & $7$ & $9$ & $11$\\
\midrule
$k^\ast(\ell)$ & $\mathbf{1}$ & $\mathbf{3}$ & $4$--$5$ & $5$--$7$ & $6$--$9$ & $7$--$11$\\
\bottomrule
\end{tabular}
\end{table}

\begin{conjecture}
\label{conj}
$k^\ast(\ell)=\ell+1$ for every $\ell\ge2$.
\end{conjecture}

\section{Verification}
\label{sec:exp}

Because several statements replace a continuous optimisation by a combinatorial
formula, we implemented independent checks in exact rational arithmetic.

\paragraph{The decomposition against blind brute force.} We implemented the game
definition literally --- for given $P$ and $Q$, count the voters with
$\dist(v,Q)<\dist(v,P)$ --- and maximised it over all $\ell$-subsets of a dense
rational grid, with no reference to Theorem~\ref{thm:decomp}. On $175$ random
instances ($n\le9$, $k,\ell\le3$, grid step $\tfrac16$ refined to $\tfrac1{24}$
where needed) the two values agreed in every case. The block form of $h$
(Lemma~\ref{lem:block}) was checked against a direct sliding-window computation
on $4000$ random gaps: no discrepancy.

\paragraph{The $\ell=1$ algorithm.} Theorem~\ref{thm:ell1} was compared with
exhaustive maximisation over a grid of resolution $\tfrac18$ of the voter
spacing on $500$ random weighted instances (up to $6$ positions, multiplicities
$\le3$, $k\le4$): no discrepancy. An earlier version of the greedy, which
neutralised a heavy block only by narrowing the window and not by excluding it
from the gap, disagreed on $194$ of $400$ instances; the check has real
discriminating power.

\paragraph{The facility advantage.} We searched instances on $m$ consecutive
integer positions with integer multiplicities (a multiplicity $0$ is an empty
site, so spacings need not be uniform), evaluating $\Gamma_{k,\ell}$ exactly
for each instance by Theorem~\ref{thm:decomp} over all $P\subset\tfrac12\mathbb{Z}$,
which suffices by Lemma~\ref{lem:disc}. The search found instances defeating
$k=\ell=2$ (smallest at $m=6$, $n=9$) and $k=\ell=3$ ($m=8$, $n=13$), and
\emph{no} instance defeating $k=\ell+1$ in any of the following ranges:
$k=3,\ell=2$ with $m\le8$, multiplicities $\le5$, $n\le16$; $k=4,\ell=3$ with
$m\le9$, multiplicities $\le4$, $n\le16$; $k=4,\ell=3$ with $m\le10$,
multiplicities in $\{0,\dots,3\}$, $n\le15$; and $k=5,\ell=4$ with $m\le10$,
multiplicities $\le3$, $n\le17$ --- $238{,}185$ instances in all, up to mirror
symmetry. The family $F(\ell)$ was confirmed for $\ell=2,\dots,6$.
The conclusion of Proposition~\ref{prop:uniform} was confirmed without the
divisibility hypothesis for $m\le30$, $\ell\le5$, $\ell\le k\le\ell+1$.

We stress what a search can and cannot establish: an enumeration over a grid of
strategies for $\mathcal{P}$ gives a \emph{lower} bound on
$\Gamma_{k,\ell}=\max_P\min_Q$, so it certifies that $\mathcal{P}$ wins but never,
by itself, that $\mathcal{Q}$ does. This is why Lemma~\ref{lem:disc} is
load-bearing in Theorem~\ref{thm:lower}, and why the negative results above are
reported as search outcomes rather than theorems.

\section{Concluding Remarks}
\label{sec:concl}

The one-round discrete Voronoi game on a line and the absolute
$(\ell|k)$-centroid on a path are the same problem, and the two literatures
have each solved parts of it: the responder is transparent (Spoerhase and Wirth,
Megiddo et al.), the leader with a single opponent facility is near-linear
(Section~\ref{sec:ell1}), and the general leader problem is polynomial in the
unary demand (de Berg et al.) but NP-hard in binary demand (Spoerhase and
Wirth). Against this background the ``simpler algorithm'' asked for
in~\cite{deberg2019oneround} is really a question about discretising the
leader's continuous positions polynomially in the unary model --- exactly what
Spoerhase and Wirth conjecture to be impossible in the binary model.

Three questions remain. (1) Is Conjecture~\ref{conj} true? Our upper bound
discards the concavity in Lemma~\ref{lem:clear} entirely, which is where the
factor of two is lost. (2) Does the one-pass sweep of Theorem~\ref{thm:ell1}
extend to $\ell=2$, where the responder may also clear a gap with two facilities?
(3) Does the structure survive on a cycle, where the two outer gaps merge into a
single wrap-around gap?

\paragraph{Use of AI tools.} The author used an AI assistant (Claude,
Anthropic) in preparing this paper: for literature search, for drafting the
text, for developing and checking proofs, and for writing the verification code
of Section~\ref{sec:exp}. The author checked every result, proof and citation,
and takes full responsibility for the content.

\appendix
\section{Proofs for Section~\ref{sec:structure}}
\label{app:proofs}

\begin{proof}[of Lemma~\ref{lem:locality}]
Assume $q\in G_i=(p_i,p_{i+1})$; the outer cases are analogous. If $v\le p_i$
then $\dist(v,P)\le p_i-v<q-v=|v-q|$, so $v\notin W(q)$; symmetrically for
$v\ge p_{i+1}$. Hence $W(q)\subseteq G_i$. For $v\in G_i$,
$\dist(v,P)=\min(v-p_i,\;p_{i+1}-v)$, so $v\in W(q)$ iff $|v-q|<v-p_i$ and
$|v-q|<p_{i+1}-v$. If $v\ge q$ the first inequality reads $p_i<q$, which holds;
if $v<q$ it reads $v>\tfrac{p_i+q}{2}$. Symmetrically the second is vacuous for
$v\le q$ and equivalent to $v<\tfrac{q+p_{i+1}}{2}$ for $v>q$. \qed
\end{proof}

\begin{proof}[of Lemma~\ref{lem:window}]
By Lemma~\ref{lem:locality}(c), $W(q)=V\cap I_q$ with
$I_q=\bigl(\tfrac{a+q}{2},\tfrac{a+q}{2}+\lambda\bigr)$. The map
$q\mapsto\tfrac{a+q}{2}$ is a bijection from $(a,b)$ onto $(a,a+\lambda)$, and
$I=(x,x+\lambda)$ satisfies $I\subseteq(a,b)$ iff $x\in[a,a+\lambda]$. This
proves the first claim and the equality case of the second. For $x=a$, an
admissible $x'>a$ close to $a$ has no voter in $(a,x']$ since $V$ is finite, so
$V\cap(x',x'+\lambda)\supseteq V\cap(a,a+\lambda)$; $x=a+\lambda$ is symmetric.
(Equality can fail only if a voter sits exactly at $a+\lambda$.) \qed
\end{proof}

\begin{proof}[of Lemma~\ref{lem:block}]
Every $W(q)$ is a block of span $<\lambda$ by Lemma~\ref{lem:window}, so
$\mathcal{Q}$ wins at most $h(G)$. Conversely let $B$ be a block with
$\spn(B)<\lambda$. An open interval $(x,x+\lambda)$ contains $B$ iff
$\max B-\lambda<x<\min B$, and lies in $(a,b)$ iff $a\le x\le a+\lambda$. The
open interval $(\max B-\lambda,\min B)$ is non-empty because $\spn(B)<\lambda$,
and it meets $[a,a+\lambda]$ because $\min B>a$ and $\max B<b=a+2\lambda$. For
such an $x$, Lemma~\ref{lem:window} gives $q$ with
$W(q)\supseteq V\cap(x,x+\lambda)\supseteq B$. \qed
\end{proof}

\begin{proof}[of Lemma~\ref{lem:clear}]
For $G=G_0$ take $q=p_1-\delta$: $W(q)=V\cap(-\infty,p_1-\delta/2)=V_{G_0}$ for
small $\delta$. For an inner gap $(a,b)$ take $q=a+\delta$, $q'=b-\delta$; their
windows $(a+\tfrac{\delta}{2},\tfrac{a+b+\delta}{2})$ and
$(\tfrac{a+b-\delta}{2},b-\tfrac{\delta}{2})$ cover $V_G$ for small $\delta$.
For concavity let $c=\tfrac12(a+b)$ and split $V_G$ into $A:=V\cap(a,c)$,
$M:=V\cap\{c\}$, $B:=V\cap(c,b)$. Both $A\subseteq(a,c)$ and
$M\cup B\subseteq[c,b)$ are blocks of span $<\lambda$, so by
Lemma~\ref{lem:block} $h(G)\ge\max\{\wt(A),\wt(M)+\wt(B)\}\ge\tfrac12\wt(V_G)$.
\qed
\end{proof}

\begin{proof}[of Theorem~\ref{thm:decomp}]
By Lemma~\ref{lem:locality} the facilities act independently on distinct gaps,
so $Q$ is described by the numbers $t_G$ of facilities per gap, $\sum t_G=\ell$,
with payoff $\sum_G\sigma_G(t_G)$ where $\sigma_G(0)=0$, $\sigma_G(t)=\wt(V_G)$
for $t\ge1$ if $G$ is outer, and $\sigma_G(1)=h(G)$, $\sigma_G(t)=\wt(V_G)$ for
$t\ge2$ if $G$ is inner (Lemmas~\ref{lem:block}, \ref{lem:clear}). The
successive increments of $\sigma_G$ are the entries $G$ contributes to $M(P)$,
followed by zeros, and by Lemma~\ref{lem:clear} they are non-increasing. A sum
of separable concave functions under a cardinality constraint is maximised by
greedily taking the $\ell$ largest increments, which is the claim. For the
running time, one merge of $V$ against $P$ yields all $\wt(V_G)$; each $h(G)$
is a two-pointer scan over $V_G$ maintaining the block weight by prefix sums;
the $\ell$ largest of $\le2k$ numbers are found by linear-time selection. \qed
\end{proof}

\end{document}